\documentclass{IEEEtran}
\usepackage{graphicx}
\usepackage{mathrsfs}
\usepackage{stfloats}
\usepackage{dsfont}
\usepackage{setspace}
\usepackage{array}
\usepackage{bbm}
\usepackage[bookmarks,bookmarksopen,bookmarksdepth=2]{hyperref}
\usepackage{here}
\usepackage{amsmath,amsthm,amsfonts,amssymb}
\usepackage{color}
\usepackage{epstopdf}
\usepackage[centerlast,small]{caption}
\usepackage{multicol}
\newtheorem{definition}{Definition}
\newtheorem{theorem}{Theorem}

\newtheorem{proposition}{Proposition}

\newtheorem{remark}{Remark}

\def\S{\mathbf{A}}

\def\0{\mathbf{0}}

\def\S{\mathcal{S}}

\begin{document}
\title{\huge Analytical Modeling of the Path-Loss for Reconfigurable Intelligent Surfaces -- Anomalous Mirror or Scatterer ?}
\author{\normalsize{{Marco Di Renzo$^{(1)}$, Fadil Habibi Danufane$^{(1)}$, Xiaojun Xi$^{(1)}$, Julien de Rosny$^{(2)}$, and Sergei Tretyakov$^{(3)}$}\\
{{$^{(1)}$ Universit\'e Paris-Saclay, CNRS, CentraleSup\'elec, Laboratoire des Signaux et Syst\`emes, Gif-sur-Yvette, France} \\
{$^{(2)}$ ESPCI, Institut Langevin, CNRS, Paris, France} \\
{$^{(3)}$ Aalto University, Dept. of Electronics and Nanoengineering, Helsinki, Finland} \\
{email: marco.direnzo@centralesupelec.fr}} \vspace{-0.8cm}  }}
\maketitle

\begin{abstract}
Reconfigurable intelligent surfaces (RISs) are an emerging field of research in wireless communications. A fundamental component for analyzing and optimizing RIS-empowered wireless networks is the development of simple but sufficiently accurate models for the power scattered by an RIS. By leveraging the general scalar theory of diffraction and the Huygens-Fresnel principle, we introduce simple formulas for the electric field scattered by an RIS that is modeled as a sheet of electromagnetic material of negligible thickness. The proposed approach allows us to identify the conditions under which an RIS of finite size can or cannot be approximated as an anomalous mirror. Numerical results are illustrated to confirm the proposed approach.
\end{abstract}

\vspace{-0.4cm}
\section{Introduction}
In contemporary wireless networks, transmitters and receivers can be programmed and controlled for optimizing the system performance. The environmental objects (buildings, walls, ceilings, etc.) that constitute the wireless environment cannot, on the other hand, be customized based on the network conditions. This status quo has recently been challenged by the emerging technology of reconfigurable intelligent surfaces (RISs) -- Thin sheets of electromagnetic materials that are capable of shaping the radio waves in arbitrary ways \cite{Marco-Smart-radio-environments}, \cite{Zappone-Wireless-Network-Design}. The overarching vision consists of coating the environmental objects with RISs and optimizing their properties, in order to, e.g., reflect an impinging radio wave towards a desired direction with the objective of capitalizing from multipath propagation rather than being negatively affected by it \cite{Liaskos2}-\cite{Holographic}.

In simple terms, an RIS is the two-dimensional equivalent of a reconfigurable meta-material, and is made of elementary elements called scattering particles or meta-atoms \cite{Yu-Light-propagation}. Depending on the arrangement and configuration of the scattering particles, an RIS is capable of altering the wavefront of the radio waves impinging upon it. For example, RISs can modify the direction of the reflected or refracted waves and their polarization, or can encode data onto the shape of the scattered waves \cite{Asadchy-Perfect-control,Estakhri-Wave-front,Liu-Intelligent}. The two-dimensional nature of RISs make them easier to design, less lossy, less expensive, and easier to deploy than their three-dimensional counterpart. Broadly speaking, RISs are special surfaces that are engineered to possess properties that cannot be found in surfaces made of naturally occurring materials \cite{Liu-Intelligent}. Thanks to these properties, RISs are receiving major attention from the wireless community, and are considered to be the key enabler of the emerging concept of smart radio environments \cite{Marco-Smart-radio-environments}.

A major open research issue for analyzing the ultimate performance limits, optimizing the operation, and assessing the advantages and limitations of RIS-empowered wireless networks is the development of simple but sufficiently accurate models for the power received at a given location in space when a transmitter emits radio waves that illuminate an RIS. Recently, a few research works have tackled this research issue. In \cite{Tang-RIS}, the authors have performed a measurement campaign in an anechoic chamber and have shown that the power reflected from an RIS follows a scaling law that depends on many parameters, including the size of the RIS, the mutual distances between the transmitter/receiver and the RIS (i.e., near-field vs. far-field conditions), and whether the RIS is used for beamforming or broadcasting applications. In \cite{Bucheli-RIS-bridging}, the authors have employed antenna theory to compute the electric field in the near-field and far-field of a finite-size RIS, and have proved that an RIS is capable of acting as an anomalous mirror in the near-field of the array. The results are obtained numerically and no explicit analytical formulation of the received power as a function of the distance is given. Similar results have been obtained in \cite{Ellingson-Path-loss}. In \cite{Khawaja-Coverage}, the power measured from passive reflectors in the millimeter-wave frequency band is compared against ray tracing simulations. By optimizing the area of the surface that is illuminated, it is shown that a finite-size passive reflector can act as an anomalous mirror. The study in \cite{Emil} relies on the assumption of plane waves and is valid only in the far-field of the RIS (i.e., long distances).

In this paper, we leverage the general scalar theory of diffraction and the Huygens-Fresnel principle, and introduce simple closed-form expressions to compute the power reflected from an RIS as a function of the distance between the transmitter/receiver and the RIS, the size of the RIS, and the phase transformation applied by the RIS. With the aid of the stationary phase method, we identify sufficient conditions under which an RIS acts as an anomalous mirror, and, therefore, the received power decays as a function of the reciprocal of the sum of the distances between the transmitter and the RIS, and the RIS and the receiver. For simplicity, the analytical formulas are reported without proof and for a one-dimensional RIS. The proofs, discussions on the boundary conditions to solve Maxwell's equations, the impact of the direct link, the analysis of refraction, and two-dimensional RISs can be found in the companion journal paper \cite{Danufane-TBA}.

The rest of this paper is organized as follows. In Section II, the system model is introduced. In Section III, the analytical formulation of the electric field emitted by a point source and scattered by a finite-size RIS is reported. Explicit expressions of the electric field in the near-field and far-field are given. In Section IV, numerical results are provided to illustrate the scaling laws of the received power as a function of the transmission distances. Finally, Section V concludes this paper.

\section{System Model}
In a two-dimensional space, we consider a system that consists of a transmitter (Tx), a receiver (Rx), and a flat surface ($\S$) of zero-thickness. Without loss of generality, we assume that $\S$ is located such that its center coincides with the origin. Furthermore, $\S$ lies in the $x$-axis and spans along $[-L,L]$, i.e., $\S = \left\{(x,0): -L \leq x \leq L\right\}$. In other words, $\S$ is a straight line. The locations of Tx and Rx are denoted by $(x_T,y_T)$ and $(x_R,y_R)$, respectively. We consider only the scenario where Tx and Rx are on the same side of the surface $\S$, i.e., we focus our attention on modeling reflections from the surface $\S$. Therefore, $y_T$ and $y_R$ take positive values, while there is no restriction on the values taken by $x_T$ and $x_R$.

Tx is modeled as a point source that emits cylindrical electromagnetic (EM) waves through the vacuum. The EM waves emitted by Tx travel at the speed of light $c$. The frequency of the EM waves is denoted by $f$, and the wavelength and wavenumber are $\lambda = c/f$ and $k = 2\pi /\lambda$, respectively. We are interested in computing the intensity of the electric field emitted by Tx and observed at an arbitrary point, i.e. Rx, on the positive $y$-axis, with the exception of the location of the point source. In vacuum, the $x$ and $y$ components of the electric field are not coupled, and we assume that $\S$ does not change the polarization of the EM waves. Under these assumptions, we can analyze any components of the electric field. We consider the tangential (to the surface $\S$) component of the electric field, which is denoted by ${E_x}\left( {{x_R},{y_R}} \right)$.

For every point $(x,0) \in \S$, the Tx-to-$\S$ and $\S$-to-Tx distances are denoted by ${d_T}\left( x \right) = \sqrt {{{\left( {x - {x_T}} \right)}^2} + y_T^2}$ and ${d_R}\left( x \right) = \sqrt {{{\left( {x_R - x} \right)}^2} + y_R^2}$, respectively. In particular, ${d_T}\left( x \right)$ is the radius of the wavefront of the EM wave that is emitted by Tx and intersects $\S$ at $(x,0)$, and ${d_R}\left( x \right)$ is the radius of the wavefront of the EM wave that originates from $\S$ at $(x,0)$ and is observed at Rx. With a similar terminology, the angle of incidence of the EM wave at $(x,0) \in \S$ is denoted by $\theta_T(x)$. It represents the angle formed by the $y$-axis and the wavefront of the EM wave that originates from Tx and intersects $\S$ at $(x,0)$. The angle of reflection of the EM wave at $(x,0) \in \S$ is denoted by $\theta_R(x)$, and it represents the angle formed by the $y$-axis and the wavefront of the EM wave that is emitted by $\S$ at $(x,0)$ and is observed at Rx.

For simplicity, we assume ${d_T}\left( x \right) \gg \lambda$ and ${d_R}\left( x \right) \gg \lambda$, which usually hold true in practical setups \cite{Tang-RIS}. The complete analysis is available in \cite{Danufane-TBA}. Under these assumptions, the electric field emitted by the point source (Tx) and observed at Rx in the absence of $\S$ corresponds to the Green function in the plane, which is well approximated as follows \cite{Book}: 
\begin{equation}
{E_x}\left( {{x_R},{y_R}} \right) \approx {E_0}\frac{{\exp \left( { - jk{d_{TR}}\left( {{x_R},{y_R}} \right)} \right)}}{{\sqrt {k{d_{TR}}\left( {{x_R},{y_R}} \right)} }}
\end{equation} 
\noindent where ${E_0} =  - j\sqrt {{1 \mathord{\left/ {\vphantom {1 {\left( {8\pi } \right)}}} \right. \kern-\nulldelimiterspace} {\left( {8\pi } \right)}}} \exp \left( { - {{j\pi } \mathord{\left/ {\vphantom {{j\pi } 4}} \right. \kern-\nulldelimiterspace} 4}} \right)$, $j$ is the imaginary unit, and ${d_{TR}}\left( {{x_R},{y_R}} \right) = \sqrt {{{\left( {{x_R} - {x_T}} \right)}^2} + {{\left( {{y_R} - {y_T}} \right)}^2}}$ is the  distance between Tx and Rx. 

The surface $\S$ is modeled as a spatially-inhomogeneous reflector that is capable of modifying the phase of the incident field. We assume that the electromagnetic properties of the surface $\S$ vary slowly, as compared with the wavelength, along the surface itself. Under this approximation, the surface $\S$ can be well modeled as a local structure: the reflected field at $(x,0) \in \S$ depends, approximately, only on the incident field at $(x,0) \in \S$ \cite{Caloz}. More precisely, the reflection coefficient at $(x,0) \in \S$ can be written as follows:
\begin{equation}
\Gamma_r(x) = C(x)\exp\left(j\Phi(x)\right)
\end{equation}
\noindent where $C(x) \in \mathbb{R}^+$ and $\Phi(x) \in [0,2\pi)$ denote the amplitude and phase of the reflection coefficient, respectively. In this paper, we are interested in analyzing only reflections. Furthermore, we assume that the surface $\S$ operates in the regime of a phase-gradient reflector and, therefore, assume $C(x)=1$.

\section{Electric Field Reflected from $\S$}
Based on the assumptions in Section II, the intensity of the electric field emitted by Tx, reflected by the surface $\S$, and observed at $\left( {{x_R},{y_R}} \right)$ for ${y_R} > 0$ is given as follows.
\begin{theorem}\label{Theorem_1}
Let us assume ${d_T}\left( x \right) \gg \lambda$ and ${d_R}\left( x \right) \gg \lambda$. The electric field ${E_x}\left( {{x_R},{y_R}} \right)$ can be formulated as follows:
\begin{equation} \label{Eq_Efield}
{{E_x}\left( {{x_R},{y_R}} \right)} = {{\mathcal{I}}_0}\int\nolimits_{ - L}^{ + L} {{\mathcal{I}}\left( x \right)\exp \left( { - jk{\mathcal{P}}\left( x \right)} \right)dx}
\end{equation}
\noindent where ${{\mathcal{I}}_0} = {1 \mathord{\left/{\vphantom {1 {\left( {8\pi } \right)}}} \right. \kern-\nulldelimiterspace} {\left( {8\pi } \right)}}$ and:
\begin{equation}
{\mathcal{P}}\left( x \right) = {d_T}\left( x \right) + {d_R}\left( x \right) - \Phi \left( x \right)
\end{equation}
\begin{equation}
{\mathcal{I}}\left( x \right) = \frac{1}{{\sqrt {{d_T}\left( x \right){d_R}\left( x \right)} }}\left( {\frac{{{y_T}}}{{{d_T}\left( x \right)}} + \frac{{{y_R}}}{{{d_R}\left( x \right)}}} \right)
\end{equation}
\end{theorem}
\begin{proof}
It follows by formulating in mathematical terms the Huygens-Fresnel principle by using the general scalar theory of diffraction and by applying appropriate boundary conditions at the surface $\S$. The details can be found in \cite{Danufane-TBA}.
\end{proof}

The electric field in \eqref{Eq_Efield} is formulated in a simple integral form, which, however, does not explicitly unveil the dependency of the electric field as a function of the transmission distances. Also, the electric field depends on the specific phase shift $\Phi(x)$ applied by the surface $\S$. In the following three sub-sections, we consider three case studies for the choice of $\Phi(x)$. Due to space limitations, only the first two case studies are analyzed in detail. For both cases, we introduce explicit approximate closed-form expressions for the electric field in \eqref{Eq_Efield} for short and long transmission distances. The definition of long transmission distance is given as follows.
\begin{definition}\label{Definition}
Let us define ${d_Q}\left( x \right) = \sqrt {{{\left( {{x_Q} - x} \right)}^2} + y_Q^2}$ and $\sin \left( {{\theta _{Q0}}} \right) = \sin \left( {{\theta _{Q}}\left( 0 \right)} \right) = {{ - q{x_Q}} \mathord{\left/ {\vphantom {{ - q{x_Q}} {{d_{Q0}}}}} \right. \kern-\nulldelimiterspace} {{d_{Q0}}}}$ for $Q=\{T, R \}$, where $q=1$ if $Q=T$ and $q=-1$ if $Q=R$, 	as well as ${d_{Q0}} = \sqrt {x_Q^2 + y_Q^2}$ for $Q=\{T, R \}$. The system is said to operate in the long distance regime if the approximation ${d_Q}\left( x \right) \approx {d_{Q0}} + qx\sin \left( {{\theta _{Q0}}} \right)$ holds true for $Q=\{T, R \}$. Otherwise, it is said to operate in the short distance regime.
\end{definition}

Loosely speaking (with a slightly abuse of terminology), the long and short distance operating regimes can be identified with the far-field and near-field regimes, respectively.

\subsection{$\S$ Acting as a Uniform Reflecting Surface}
In this section, we analyze the case study in which the surface $\S$ operates as a mirror reflector. This operation is obtained by setting $\Phi \left( x \right)=\phi _0$ for $x \in \S$ in \eqref{Eq_Efield}, where $\phi _0 \in [0,2\pi)$ is a fixed phase shift. The following two propositions report approximate closed-form expressions of the intensity of the electric field in \eqref{Eq_Efield} under the assumption of short and long distance regime, respectively.
\begin{proposition} \label{Proposition_1}
In the short distance regime, the intensity of the electric field in \eqref{Eq_Efield} can be approximated as follows:
\begin{equation} \label{Efield_noRISmirror}
\left| {{E_x}\left( {{x_R},{y_R}} \right)} \right| \approx \frac{1}{{\sqrt {8\pi k} }}\frac{1}{{\sqrt {{d_T}\left( {{x_s}} \right) + {d_R}\left( {{x_s}} \right)} }}
\end{equation}
\noindent where ${x_s} \in \left[ {-L,L} \right]$ is the unique solution of the equation:
\begin{equation} \label{Efield_noRISmirror_Xs}
\frac{{{x_s} - {x_T}}}{{{d_T}\left( {{x_s}} \right)}} - \frac{{{x_R} - {x_s}}}{{{d_R}\left( {{x_s}} \right)}} = 0
\end{equation}
\end{proposition}
\begin{proof}
It follows from \eqref{Eq_Efield} by applying the stationary phase method. The details can be found in \cite{Danufane-TBA}.
\end{proof}
\begin{remark}
Proposition \ref{Proposition_1} holds true only if \eqref{Efield_noRISmirror_Xs} has at least one solution ${x_s} \in \left[ {-L,L} \right]$. The case study when this does not hold true  can be found in \cite{Danufane-TBA}. Similar comments hold for similar case studies analyzed in the following sub-sections.
\end{remark}

From Proposition \ref{Proposition_1}, the following conclusions follow.
\begin{itemize}
\item In the short distance regime, the surface $\S$ behaves as a specular mirror. In particular, the (end-to-end) intensity of the electric field reflected from the surface decays as a function of the reciprocal of the square root of the sum of the Tx-to-$\S$ and $\S$-to-Rx distances. The presence of the square root originates from the assumption of two-dimensional space (see the emitted field in (1)).
\item Equation \eqref{Efield_noRISmirror} can be regarded as an approximation of \eqref{Eq_Efield} under the condition of geometric optics propagation. More precisely, \eqref{Efield_noRISmirror} unveils that the intensity of the electric field is approximately the same as that obtained from a single ray (i.e., the direction of propagation of the wavefront of the EM wave) that is obtained from the two line segments that connect Tx with the point ${x_s} \in \left[ {-L,L} \right]$ that fulfills \eqref{Efield_noRISmirror_Xs}, and the latter point with Rx. Therefore, the point ${x_s}$ can be referred to as reflection point.
\item From the definition of angles of incidence and reflection, we have ${{\left( {{x_s} - {x_T}} \right)} \mathord{\left/ {\vphantom {{\left( {{x_s} - {x_T}} \right)} {{d_T}\left( {{x_s}} \right)}}} \right. \kern-\nulldelimiterspace} {{d_T}\left( {{x_s}} \right)}} = \sin \left( {{\theta _T}\left( {{x_s}} \right)} \right)$ and ${{\left( {{x_R} - {x_s}} \right)} \mathord{\left/ {\vphantom {{\left( {{x_R} - {x_s}} \right)} {{d_R}\left( {{x_s}} \right)}}} \right. \kern-\nulldelimiterspace} {{d_R}\left( {{x_s}} \right)}} = \sin \left( {{\theta _R}\left( {{x_s}} \right)} \right)$, respectively. From \eqref{Efield_noRISmirror_Xs}, this implies that the angles of incidence and reflection coincide at the reflection point ${x_s} \in \left[ {-L,L} \right]$. In other words, \eqref{Efield_noRISmirror} allows us to retrieve the law of reflection.
\end{itemize}
\begin{proposition} \label{Proposition_2}
In the long distance regime, the intensity of the electric field in \eqref{Eq_Efield} can be approximated as follows:
\begin{equation} \label{Efield_noRISscatterer}
\begin{split}
\left| {{E_x}\left( {{x_R},{y_R}} \right)} \right| &\approx \frac{1}{{4\pi }}\left| \frac{{\cos \left( {{\theta _{T0}}} \right) + \cos \left( {{\theta _{R0}}} \right)}}{{\sqrt {{d_{T0}}{d_{R0}}} }}\right|\\
& \times \left|\frac{{\sin \left( {kL\left( {\sin \left( {{\theta _{T0}}} \right) - \sin \left( {{\theta _{R0}}} \right)} \right)} \right)}}{{k\left( {\sin \left( {{\theta _{T0}}} \right) - \sin \left( {{\theta _{R0}}} \right)} \right)}}\right|
\end{split}
\end{equation}
\noindent where $\cos \left( {{\theta _{Q0}}} \right) = {{{y_Q}} \mathord{\left/ {\vphantom {{{y_Q}} {{d_{Q0}}}}} \right. \kern-\nulldelimiterspace} {{d_{Q0}}}}$ for $Q=\{T, R \}$.
\end{proposition}
\begin{proof}
It follows by using the approximation ${d_Q}\left( x \right) \approx {d_{Q0}} + qx\sin \left( {{\theta _{Q0}}} \right)$ for $Q=\{T, R \}$. The details are in \cite{Danufane-TBA}.
\end{proof}

From Proposition \ref{Proposition_2}, we evince that the surface $\S$ does not behave as a specular mirror for long transmission distances. The scaling law that governs the intensity of the electric field as a function of the distances is, in addition, no straightforward to be identified. To shed light on the impact of the transmission distances, i.e., ${d_{T0}}$ and ${d_{R0}}$ in \eqref{Efield_noRISscatterer}, we consider the case study in which Tx and Rx move along two straight lines such that the angles $\theta_{T0}$ and $\theta_{R0}$ are kept constant (the two angles do not have to be necessarily the same), but ${d_{T0}}$ and ${d_{R0}}$ are different. In this case, we observe from \eqref{Efield_noRISscatterer} that the intensity of the electric field decays as a function of the square root of the product of the distances between Tx and the (center of the) surface $\S$, and the (center of the) surface $\S$ and Rx. In this regime, therefore, the surface $\S$ is better modeled as a scatterer, since its size is relatively small in comparison with the transmission distances involved (i.e., ${d_{T0}}$ and ${d_{R0}}$).

The findings in Propositions \ref{Proposition_1} and \ref{Proposition_2} provide us with evidence that justifies the validity of Theorem \ref{Theorem_1}, and, therefore, substantiate the approach embraced in this paper for modeling and analyzing RISs in wireless networks. This constitutes the departing point of the following two sub-sections.

\subsection{$\S$ Acting as a Reconfigurable Intelligent Surface}
In this section, we analyze the case study in which the surface $\S$ operates as an RIS whose phase $\Phi \left( x \right)$ can be appropriately optimized. In particular, we assume that $\S$ acts as an anomalous reflector that is configured for reflecting the EM waves emitted by Tx towards a given direction. Due to the assumption $C(x)=1$, we implicitly ignore parasitic scattering. To this end, $\Phi \left( x \right)$ in \eqref{Eq_Efield} is chosen as follows:
\begin{equation} \label{Phi_RIS}
\Phi \left( x \right) = \left( {{{\bar \phi }_T} - {{\bar \phi }_R}} \right)x + {\phi _0}/k
\end{equation}
\noindent where $\phi _0 \in [0,2\pi)$ is a fixed phase shift, ${{\bar \phi }_T} =  - {{{{\bar x}_T}} \mathord{\left/ {\vphantom {{{{\bar x}_T}} {\sqrt {\bar x_T^2 + \bar y_T^2} }}} \right. \kern-\nulldelimiterspace} {\sqrt {\bar x_T^2 + \bar y_T^2} }}$, ${{\bar \phi }_R} = {{{{\bar x}_R}} \mathord{\left/ {\vphantom {{{{\bar x}_R}} {\sqrt {\bar x_R^2 + \bar y_R^2} }}} \right. \kern-\nulldelimiterspace} {\sqrt {\bar x_R^2 + \bar y_R^2} }}$, and $\left( {{{\bar x}_T},{{\bar y}_T}} \right)$ and $\left( {{{\bar x}_R},{{\bar y}_R}} \right)$ are parameters that are optimized for obtaining the desired reflection capability, as detailed in further text. 

In contrast with Section III-A, this case study needs more elaboration to unveil the scaling law of the intensity of the electric field as a function of the distances. To this end, we consider the specific setup in which ${{\bar \phi }_T} = \sin \left( {{\theta _T}\left( 0 \right)} \right)$ and ${{\bar \phi }_R} = \sin \left( {{\theta _R}\left( 0 \right)} \right)$, which corresponds to a surface $\S$ that is configured by taking into account the angles of incidence and reflection of the EM with respect to $(0,0)$. Other case studies are analyzed in \cite{Danufane-TBA}. It is worth mentioning that this setup does not necessarily imply $\left( {{{\bar x}_T},{{\bar y}_T}} \right) = \left( {{x_T},{y_T}} \right)$ and $\left( {{{\bar x}_R},{{\bar y}_R}} \right) = \left( {{x_R},{y_R}} \right)$, which would imply that the locations of Tx (the point source) and Rx (the observation point) need to be exactly known. Setups corresponding to different locations of Tx and Rx, but yielding the same angles of incidence and reflection, are included in the considered case study. For example, Tx and Rx move along two straight lines in which the angles with the $y$-axis at $(0,0)$ are kept constant but the distances are not. 

The following two propositions report approximate expressions of the intensity of the electric field in \eqref{Eq_Efield} under the assumption of short and long distance regime, respectively.
\begin{proposition} \label{Proposition_3}
Assume ${{\bar \phi }_T} = \sin \left( {{\theta _T}\left( 0 \right)} \right)$ and ${{\bar \phi }_R} = \sin \left( {{\theta _R}\left( 0 \right)} \right)$. In the short distance regime, the intensity of the electric field in \eqref{Eq_Efield} can be approximated as follows:
\begin{equation} \label{Efield_RISmirror__Distance}
\left| {{E_x}\left( {{x_R},{y_R}} \right)} \right| \approx \frac{1}{{4\sqrt {2\pi k} }}\frac{{\sqrt {1 - \bar \phi _T^2}  + \sqrt {1 - \bar \phi _R^2} }}{{\sqrt {\left( {1 - \bar \phi _R^2} \right){d_{T0}} + \left( {1 - \bar \phi _T^2} \right){d_{R0}}} }}
\end{equation}
\end{proposition}
\begin{proof}
The proof is similar to Proposition \ref{Proposition_1} \cite{Danufane-TBA}.
\end{proof}
\begin{proposition} \label{Proposition_4}
Assume ${{\bar \phi }_T} = \sin \left( {{\theta _T}\left( 0 \right)} \right)$ and ${{\bar \phi }_R} = \sin \left( {{\theta _R}\left( 0 \right)} \right)$. In the long distance regime, the intensity of the electric field in \eqref{Eq_Efield} can be approximated as follows:
\begin{equation} \label{Efield_RISscatterer__Distance}
\left| {{E_x}\left( {{x_R},{y_R}} \right)} \right| \approx \frac{L}{{4\pi }}\frac{{\sqrt {1 - \bar \phi _T^2}  + \sqrt {1 - \bar \phi _R^2} }}{{\sqrt {{d_{T0}}{d_{R0}}} }}
\end{equation}
\end{proposition}
\begin{proof}
The proof is similar to Proposition \ref{Proposition_2} \cite{Danufane-TBA}.
\end{proof}

Since ${{\bar \phi }_T}$ and ${{\bar \phi }_R}$ in Propositions \ref{Proposition_3} and \ref{Proposition_4} do not depend on the distances, the following conclusions can be drawn.
\begin{itemize}
\item From \eqref{Efield_RISmirror__Distance}, RISs behave as anomalous mirrors in the short distance regime: the intensity of the electric field decays with the square root of a weighted sum of the distances, but the angles of incidence and reflection can be different.
\item From \eqref{Efield_RISscatterer__Distance}, RISs behave as scatterers in the long distance regime: the intensity of the electric field decays as a function of the square root of the product of the distances.
\item In \eqref{Efield_RISmirror__Distance} and \eqref{Efield_RISscatterer__Distance}, ${{\bar \phi }_T}$ is configured based on the direction of incidence (at $(0,0)$ and with the $y$-axis) of the EM wave emitted by Tx, and ${{\bar \phi }_R}$ is configured based on the desired direction of reflection (at $(0,0)$ and with the $y$-axis) of the EM wave reflected by the RIS. By optimizing ${{\bar \phi }_R}$, RISs can be configured to reflect EM waves towards, predominantly, any directions. The limitations are discussed in, e.g., \cite{Asadchy-Perfect-control}. This is different from \eqref{Efield_noRISmirror}, in which the direction of reflection and incidence coincide. This is the difference between specular and anomalous mirrors. 
\end{itemize}

\subsection{$\S$ Acting as a Passive Reflecting Beamformer}
In this section, we analyze the case study in which the surface $\S$ operates as an RIS whose phase $\Phi \left( x \right)$ is appropriately optimized in order for $\S$ to act as a beamformer. The difference with the previous sub-section can be summarized as follows.
\begin{itemize}
\item In Section III-B, the desired functionality of the RIS consists of reflecting (or steering) the incident EM wave towards a predetermined \textit{direction}. All the  receivers located in the direction of reflection benefit from the RIS. This setup is, therefore, more suitable for RISs that are employed for broadcasting applications \cite{Tang-RIS}.
\item In this sub-section, on the other hand, the desired functionality of the RIS is to \textit{focus} the EM wave towards a predetermined \textit{location}. In this case, a single or a few receivers at specific locations benefit from the RIS. 
\end{itemize}

In particular, we consider that $\S$ acts as a beamformer (or a reflecting lens) that focuses the signal towards a single location $\left( {{{\bar x}_R},{{\bar y}_R}} \right)$. To this end, $\Phi \left( x \right)$ in \eqref{Eq_Efield} is chosen as follows:
\begin{equation} \label{Phi_Beamformer}
\Phi \left( x \right) = \sqrt {{{\left( {x - {x_T}} \right)}^2} + y_T^2}  + \sqrt {{{\left( {x - {{\bar x}_R}} \right)}^2} + \bar y_R^2}
\end{equation}

From a mathematical point of view, $\Phi \left( x \right)$ in Section III-B is optimized such that the first-order derivative of $\mathcal {P} \left( x \right)$ is equal to zero at ${x_s} \in \left[ {-L,L} \right]$ (if it exists). The phase $\Phi \left( x \right)$ in \eqref{Phi_Beamformer} is, by contrast, optimized such that $\mathcal {P} \left( x \right)$ is equal to zero when evaluated at the location of interest, i.e., $\left( {{{\bar x}_R},{{\bar y}_R}} \right)$. The design criterion for optimizing the surface $\S$ is, thus, different.

The intensity of the electric field can be computed, in the long distance regime, by using analytical steps similar to those reported in Sections III-A and III-B \cite{Danufane-TBA}. Due to space limitations, the details are omitted. Numerical illustrations are, however, reported in the next section in order to showcase the difference among the three configurations for the surface $\S$.

\section{Numerical Results}
In this section, we provide some illustrative  numerical results in order to showcase the difference among the three configurations for $\S$ that are elaborated in Section III, and in order to numerically evaluate the intensity of the electric field as a function of the transmission distances. For ease of writing, the directions of incidence and reflection are identified by the angles ${{\theta _{T0}}}$ and ${{\theta _{R0}}}$, respectively. In all simulation results, we consider the following setup: (i) $L=0.75$ m; (ii) $f=28$ GHz; (iii) ${x_T} =  - {d_{T0}}\sin \left( {{\theta _{T0}}} \right)$ and ${y_T} = {d_{T0}}\cos \left( {{\theta _{T0}}} \right)$ with ${\theta _{T0}} = {\pi  \mathord{\left/ {\vphantom {\pi  4}} \right. \kern-\nulldelimiterspace} 4}$; (iv) $\left( {{{\bar x}_T},{{\bar y}_T}} \right) = \left( {{x_T},{y_T}} \right)$; and (v) $\phi_0 = 0$. This setup corresponds to a scenario in which $\S$ is employed in the millimeter-wave frequency band, and its size, $2L$, corresponds to, approximately, the diagonal of a two-dimensional surface of size 1 m$^2$. This is compatible and in agreement with other recent papers and experimental activities \cite{Tang-RIS}, \cite{Khawaja-Coverage}, \cite{DOCOMO}.

More precisely, we consider two case studies.
\begin{itemize}
\item In the first case study, we are interested in illustrating the difference among the three different configurations for $\S$ analyzed in Section III. To this end, we plot the intensity of the electric field emitted by a fixed location (Tx) and observed at different locations $\left( {{x_R},{y_R}} \right)$. The following setup is considered: (i) ${d_{T0}}=11$ m; (ii) ${{\bar x}_R} = {d_{R0}}\sin \left( {{\theta _{R0}}} \right)$ and ${{\bar y}_R} = {d_{R0}}\cos \left( {{\theta _{R0}}} \right)$ with ${d_{R0}} = 5$ m and ${\theta _{R0}} = {\pi  \mathord{\left/ {\vphantom {\pi  3}} \right. \kern-\nulldelimiterspace} 3}$. As for the surface $\S$ in Section III-B, this setup corresponds to reflecting an EM wave that is incident at an angle of $45$ degrees with the $y$-axis towards an angle of $60$ degrees with the $y$-axis. As for the surface $\S$ in Section III-C, this setup corresponds to focusing an EM wave that is incident at an angle of $45$ degrees with the $y$-axis towards the single location $\left( {{{\bar x}_R},{{\bar y}_R}} \right) = \left( {4.33,2.5} \right)$ m; and (iii) the observation region is chosen in the range ${x_R} \in \left[ { - 2,10} \right]$ m and ${y_R} \in \left[ {0,10} \right]$ m.
\item In the second case study, we are interested in illustrating the different scaling law of the intensity of the electric field as a function of the transmission distances, and, in particular, in showcasing the two operating regimes that correspond to short and long transmission distances. To this end, the following setup is considered: (i) ${x_T} =  - {d_0}\sin \left( {{\theta _{T0}}} \right)$ and ${y_T} = {d_0}\cos \left( {{\theta _{T0}}} \right)$; (ii) ${x_R} = {d_0}\sin \left( {{\theta _{R0}}} \right)$ and ${y_R} = {d_0}\cos \left( {{\theta _{R0}}} \right)$ with ${\theta _{R0}} = {\pi  \mathord{\left/ {\vphantom {\pi  4}} \right. \kern-\nulldelimiterspace} 4}$ for the uniform surface in Section III-A and ${\theta _{R0}} = {\pi  \mathord{\left/ {\vphantom {\pi  6}} \right. \kern-\nulldelimiterspace} 6}$ for the RIS in Section III-B; and (iii) the Tx-to-$\S$ and $\S$-to-Rx distances are the same and are in the range ${d_0} \in \left[ 0, 175 \right]$ m. Thus, the end-to-end distance is $2d_0$.
\end{itemize}
\begin{figure}[!t]
\begin{centering}
\includegraphics[width=\columnwidth]{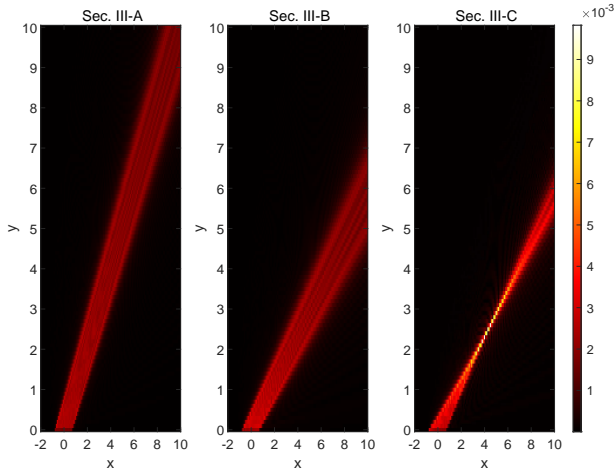}
\caption{Intensity of the electric field from Theorem \ref{Theorem_1}.}
\label{Fig_1}
\end{centering}  
\end{figure}
\begin{figure}[!t]
\begin{centering}
\includegraphics[width=\columnwidth]{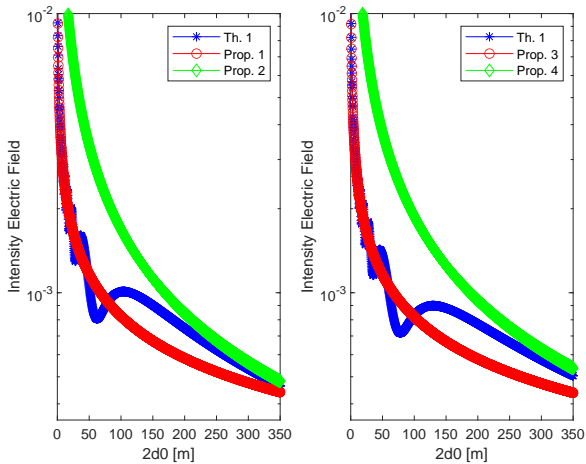}
\caption{Comparison of short and long distance approximations.}
\label{Fig_2}
\end{centering} 
\end{figure}

The results corresponding to the first case study are reported in Fig. \ref{Fig_1}, which shows the intensity of the electric field obtained from \eqref{Eq_Efield}. The figure substantiates the findings in Section III: (i) the angles of incidence and reflections of a uniform reflecting surface are the same (Section III-A); (ii) an RIS configured as described in Section III-B is capable of steering the reflected signal towards desired (anomalous) directions; and (iii) an RIS configured as described in Section III-C is capable of focusing the signal towards desired locations.

The results corresponding to the second case study are reported in Fig. \ref{Fig_2}, which compares \eqref{Eq_Efield} with the approximated closed-form expressions obtained in Sections III-A and III-B. The figure substantiates the findings in Section III. In particular: (i) the closed-form approximations for the short distance regime are accurate for end-to-end distances ($2d_0$) up to 100-150 m; and (ii) the closed-form approximations for the long distance regime are accurate for end-to-end distances ($2d_0$) greater than 200-250 m. For the considered setup, we conclude that RISs are capable of acting as anomalous mirrors for distances of the order of tens of meters. The range of distances for which the approximation holds depend, among other parameters, on the size of the surface and the operating frequency. In general, the larger the size of the surface is and the higher the operating frequency is, the more accurate the approximation as an anomalous mirror becomes, i.e., it can be used for longer transmission distances.

\section{Conclusion}
In this paper, we have leveraged the general scalar theory of diffraction in order to obtain approximate closed-form expressions of the intensity for the electric field reflected by RISs in the short and long transmission distance regimes. We have observed different scaling laws in the two considered operating regimes. The proposed approach and results constitute a first attempt to identify appropriate path-loss models for analyzing the achievable performance of RISs in wireless networks.

\section*{Acknowledgment}
This work was supported in part by the European Commission through the H2020 5GstepFWD project under grant 722429 and  the H2020 ARIADNE project under grant 871464.

\end{document}